\documentclass{article}
\usepackage{times}
\usepackage{soul}
\usepackage{url}
\usepackage{color}
\usepackage[hidelinks]{hyperref}
\usepackage[utf8]{inputenc}
\usepackage[small]{caption}
\usepackage{graphicx}
\usepackage{amsmath}
\usepackage{amsthm}
\usepackage{amsfonts}
\usepackage{booktabs}
\usepackage{algorithm}
\usepackage{algorithmic}
\usepackage{caption}
\usepackage{subcaption}
\usepackage{multirow}
\usepackage[switch]{lineno}
\usepackage[backend=biber]{biblatex}
\usepackage{authblk}  

\newtheorem{lemma}{Lemma}

\title{A Reduction-based Algorithm for the Clique Interdiction Problem \thanks{This paper has been accepted for publication at IJCAI 2025.}}

\author[1]{Chenghao Zhu\thanks{\href{mailto:axs7384@gmail.com}{axs7384@gmail.com}}}
\author[1]{Yi Zhou\thanks{Corresponding author: \href{zhou.yi@uestc.edu.cn}{zhou.yi@uestc.edu.cn}}}
\author[1]{Haoyu Jiang\thanks{\href{mailto:FirstSSAT@outlook.com}{FirstSSAT@outlook.com}}}

\affil[1]{University of Electronic Science and Technology of China}
\date{}
\begin{document}

\maketitle

\begin{abstract}
    The Clique Interdiction Problem (CIP) aims to minimize the size of the largest clique in a given graph by removing a given number of vertices.  The CIP models a special Stackelberg game and has important applications in fields such as pandemic control and terrorist identification. However, the CIP is a bilevel graph optimization problem, making it very challenging to solve. Recently, data reduction techniques have been successfully applied in many (single-level) graph optimization problems like the vertex cover problem. 
    Motivated by this, we investigate a set of novel reduction rules and design a reduction-based algorithm, RECIP, for practically solving the CIP.  RECIP enjoys an effective preprocessing procedure that systematically reduces the input graph, making the problem much easier to solve.  Extensive experiments on 124 large real-world networks demonstrate the superior performance of RECIP and validate the effectiveness of the proposed reduction rules.

\end{abstract}

\section{Introduction}
\subsection{Problem Background}
The maximum clique problem (MCP) is fundamental and well-studied in graph theory and combinatorial optimization. 
The size of the maximum clique is an important metric to measure the density or cohesion of graphs. A larger size of the maximum clique implies a denser graph structure \cite{borgatti2024analyzing}, often corresponding to tightly connected communities in applications such as social networks, biological networks, and signal processing~\cite{van2012friendship,dunbar1995social,malod2010maximum,douik2020tutorial}.

In many real-world applications, it is important not only to identify large cliques but also to disrupt or minimize their sizes through a process known as interdiction~\cite{dempe2020bilevel}. 
This leads to the notion \textbf{Clique Interdiction Problem (CIP)}.
The CIP is defined as follows: Given a graph $G$ and an interdiction budget value \( k \), decide at most \( k \) vertices to be removed from (, or interdicted in) the graph such that the size of the maximum clique in the remaining graph is minimized.
From the lens of application, the problem is used to identify critical nodes in various real-world networks. 
One can see that the vertices for interdiction are typically among the most important or influential in the graph, as their removal has the greatest impact on the graph structure. We list some specific applications in the following.

\paragraph{Pandemic Control} 
For epidemic control, identifying important nodes for disease spreading is a central topic in network epidemiology, as large cliques play a significant role in the spread of diseases \cite{wang2012impacts,vsikic2013epidemic}. Therefore, it is essential to identify critical nodes that can reduce the size of cliques and monitor these nodes to control the spread of epidemics \cite{valdez2023epidemic,grass2016two}.

\paragraph{Terrorist Identification}
Large cliques can be potential sources of catastrophic events, such as terrorist attacks or cyberattacks \cite{sageman2004understanding,berry2004emergent}. Cliques promote cohesion and solidarity, enabling large groups within terrorist or criminal networks to coordinate devastating actions. Therefore, monitoring and regulating the cohesiveness of terrorist networks is of critical importance.

\subsection{Related Literature}

The general problem of interdicting several vertices or edges such that the graph becomes less cohesive is receiving increasing attention.
Specific formulations include the clique interdiction problem in the paper, the minimum vertex blocker clique problem, which minimizes the number of vertices removed to ensure the maximum clique in the remaining graph is below a specified size \cite{nasirian2019exact,mahdavi2014minimum}, and the edge interdiction clique problem, which minimizes the maximum clique size after removing at most \( k \) edges \cite{furini2021branch,mattia2024reformulations}. 
In the optimization community, such problems are classified as bilevel optimization problems and two-stage stochastic optimization problems with recourse(2SPRs)~\cite{dempe2020bilevel}. 
In particular, the CIP models a special Stackelberg game where the leader interdicts a set of vertices and the follower maximizes the clique in the remaining graph \cite{xiao2014stackelberg}.
In this sense, the problem is different from another recently studied problem called the \( k \)-defective clique problem, which maximizes the maximum clique size after adding \( k \) edges \cite{luo2024faster,chen2021computing} because adding edges and maximizing the clique are not adversarial.

From the perspective of computational intractability, the CIP is challenging as the decision of maximum clique size in the remaining graph is  NP-hard~\cite{karp2010reducibility}, W[1]-hard~\cite{downey2012parameterized}, and hard to approximate~\cite{hastad1996clique}. 
In fact, the decision version of CIP has been shown to be \( \Sigma_2^p \)-complete~\cite{rutenburg1994propositional,grune2024completeness}. 
Nevertheless, there are at least two existing algorithms for solving the CIP practically and optimally. 
One method is to formulate the problem as a bilevel integer linear program (BILP), and then rely on the BILP solver to obtain the solution~\cite{becker2017bilevel,tang2016class,fischetti2017new}. 
However, this approach may be inefficient because it does not make full use of the specific structure of the CIP.
Another algorithm called CLINTER-INTER, probably the state-of-the-art algorithm for the CIP to our knowledge, recasts the problem to a normal single-level integer linear problem with an exponential number of rows~\cite{furini2019maximum}, then uses the maximum clique solver for constraint generation to solve this ILP.  
Clearly, the scale of the linear program depends on the input graph. 
Therefore, when the input graph is huge, which is often the case in real-world scenarios, efficient reduction pre-processing techniques can be useful for reducing the size of the linear program. This in turn improves the efficiency of the final algorithm.
On the other hand, the data reduction technique has been widely used for simplifying NP-hard problems, like vertex cover \cite{hespe2020wegotyoucovered}, independent set \cite{xiao2021efficient}, or cluster editing \cite{blasius2022branch}.
A nice recent review in \cite{abu2022recent} pointed out the possibility of extensive application of this technique for harder discrete problems.
The study of reduction rules and reduction-based algorithms for interdiction optimization problems like the basic CIP is still limited in existing literature.
Motivated by this gap, we propose an efficient algorithm for the CIP with a particular focus on data reduction, enabling better scalability and efficiency.

\subsection{Our Contributions}

Our contributions are mainly two-fold.

First, we investigate the data reduction techniques for the CIP. 
We propose novel reduction rules including \textit{color},  \textit{exact clique}, \textit{triangle}, \textit{interdiction}, and the \textit{domination} reduction rules, which are used to simplify the input.
These rules rely on an in-depth structural analysis and can reduce the input, i.e., the graph $G$ and the budget $k$, with an optimality guarantee. 

We secondly provide a whole reduction-based algorithm, RECIP, for solving the CIP in real-world graphs.
Based on the reduction rules, we provide a unified reduction algorithm that preprocesses the input instance.  Worst-case time complexity guarantees for each reduction step are also given.
We also provide polynomial-time algorithms for finding lower bounds that are tighter than current methods for some instances.
Finally, the RECIP integrates the reduction algorithm and the lower bound estimations into a branch-and-cut framework. 

Extensive experiments demonstrate that RECIP outperforms the state-of-the-art methods on nearly 90\% of real-world networks.
Notably, RECIP features a powerful graph reduction ability in preprocessing. A reduction of 85\% vertices is often observed for most graphs.
Some reduction steps which are not polynomial-time in the worst case (due to the invocation of the maximum clique oracle) perform still efficiently in practice.
The source codes are publicly available \footnote{\url{https://github.com/axs7385/RECIP}}.

\section{Preliminary}

Let $G = (V, E)$ be a simple finite undirected graph, where $V$ is the vertex set and $E$ is the edge set of $G$. 
When the context is clear, we use \( n \) to denote the number of vertices $|V|$ and \( m \) to denote the number of edges $|E|$ of the graph.
The \textit{open neighborhood} of a vertex $v \in V$ is the set of its adjacent vertices, $N_G(v) = \{u \mid \{u, v\} \in E\}$, and the size of the set is the degree of $v$, $d(v)=|N_G(v)|$. We further define the \textit{closed neighborhood} of a vertex $v \in V$ to be $N_G[v] = N(v) \cup \{v\}$. 
For convenience, we use $N(u)$ to denote $N_G(u)$ and $N[u]$ to denote $N_G[v]$ unless otherwise specified.

A vertex set \( C \subseteq V \) is called a \textit{clique} if every pair of distinct vertices \( u, v \in C \) satisfies \( \{u, v\} \in E \). 
The maximum clique size of $G$, which is the size of the largest clique in \( G \), is denoted by \( \omega(G) \). 
We denote \( T_c(n) \) as the time complexity of computing $\omega(G)$, where \( n \) is the number of vertices in $G$. 
The current computation of $\omega(G)$ is in fact very fast in large, sparse graphs empirically even though the best-known \( T_c(n) \) is still exponential $O^*(1.2^n)$ \cite{xiao2017exact} in theory. 
In the paper, we employ the well-performed algorithm in \cite{chang2019efficient} to compute $\omega(G)$ and find the maximum clique in the graphs.
Furthermore, a clique is referred to as a \textit{maximal clique} if it is not a proper subset of any other clique. 
Given a vertex set $S \subseteq V$, the subgraph of $G$ induced by $S$ is denoted by $G[S]$.

Given a graph $G=(V, E)$ and a nonnegative integer $k$,  \( \theta(G, k) \) represents the minimum value of the size of the maximum clique in the graph obtained by removing at most \( k \) vertices from \( G \). 
Formally,
\[
\theta(G,k)  = \min _{S \subseteq V, |S|\leq k} \omega(G[V \setminus S]).
\]
The set \( \mathcal{S}(G,k) \) refers to all subsets $S$ that obtain the optimal $\theta(G,k)$, i.e. $\mathcal{S}(G,k) =\{S|S=\arg\max_{S \subseteq V, |S|\leq k} \omega(G[V \setminus S])\}$, for a given $G$ and $k$.







\section{Reduction Rules}

In the paper, the reduction rule refers to the rule that reduces the size of the problem instance, for example, the number of vertices and edges of graph $G$, or the budget value \( k \), while preserving optimality.
Before introducing these reduction rules, we assume that a lower bound value $lb$ of $\theta(G,k)$ is known beforehand. 
We defer the methods of obtaining effective lower bounds to Section \ref{subsec-lowerbound}.

\paragraph{Exact Clique Reduction}
Our first reduction is based on a simple observation: If the size of the maximum clique containing vertex $u$ is less than  $\theta(G,k)$, then removing vertex $u$ does not affect the solution to this instance $G$ and $k$. 
As $\theta(G,k)$ is not known beforehand, we use the lower bound of $lb$ instead in the following reduction rule.


\begin{lemma}[Exact Clique Reduction]
\label{lemma_clique}
Given a graph $G=(V,E)$, an integer $k$, and an integer $lb$, if $\theta(G,k)\ge lb$ and there is a vertex $u\in V$ such that $\omega(G[N(u)]) \leq lb-2$, then $\theta(G,k)=\theta(G[V\setminus \{u\}])$. 
\end{lemma}
Proof of this lemma, as well as missing proofs in the remainder of the paper, are left in the appended file.
Suppose that \( \theta(G, k) \ge lb \). We can exhaustively remove all vertices \( u \in V \) that \( \omega(G[N(u)]) \leq lb - 2 \) based on this reduction rule. 
The time complexity of this reduction is \( O\left( \sum_{u \in V} T_c(d(u)) \right) \).
Specifically, we observe that the removal of a vertex from the graph does not make other vertices satisfy this removable condition by Lemma~\ref{lemma_2}. So we only need to compute the maximum clique size in $G[N(u)]$ for each $u$.

\begin{lemma}
\label{lemma_2}
    Given a graph $G=(V,E)$, an integer $k$, and two adjacent vertices $u,v \in V$, if $\omega(G[N(u)]) <\omega(G[N(v)])$, then $\omega(G[N(v)\setminus \{u\}])=\omega(G[N(v)])$.
\end{lemma}


\paragraph{Degree Reduction}
For the efficiency consideration, we hope to avoid calculating the $\omega(G[N(u)])$ for all $u\in V$ if the input graph is huge.
Hence, we propose the degree reduction rule, which is obtained from the fact that $d(u)$ is an upper bound for $\omega(G[N(u)])$. 



\begin{lemma}[Degree Reduction]
Given a graph $G=(V,E)$, an integer $k$ and an integer $lb$, if $\theta(G,k)\ge lb$ and there is a vertex $u\in V$ such that $d(u) \leq lb-2$, then $\theta(G,k)=\theta(G[V\setminus \{u\}],k)$. 
\end{lemma}
The correctness of Lemma 3 can be directly derived from Lemma 1, using the fact that \( \omega(G[N(u)]) \leq d(u) \).

The degree reduction was also used in CLIQUE-INTER \cite{furini2019maximum}, the best-known existing solver for the CIP.
Suppose that \( \theta(G, k)\ge lb \), the time complexity of exhaustively removing all vertices \( u \in V \) that \( d(u) \leq lb - 2 \) is linear time, \( O\left( m \right) \). 
This is based on the observation that, when a vertex $u$ is removed, the degree of vertices in $N(u)$ is decreased by 1, making these vertices removable as well. 
Similar to the $k$-core computation where a linear-heap is used to maintain the vertices of different degrees \cite{chang2019cohesive}, one can obtain the linear time reduction procedure.

\paragraph{Color Reduction}

A coloring of a graph \( G \) is a mapping \( c: V \to Col \) from the vertex set \( V \) to a color set \( Col \) (namely, a set of integers), such that no vertex shares the same color with any of its neighbors. 
Suppose that a coloring $c$ is given for a graph $G$. For a vertex $u$, the number of distinct colors assigned to all vertices in $N(u)$, $|\{c(v): v\in N(u)\}|$, is an upper bound on the size of the maximum clique in $G[N(u)]$. We denote this value by $ds_c(u)$ and refer to it as the \textit{saturation} of a vertex \( u \).

\begin{lemma}[Color Reduction]

Given a graph $G=(V,E)$, an integer $k$, an integer $lb$ and a coloring $c$ of $G$, if $\theta(G,k)\ge lb$ and there is a vertex $u\in V$ such that $ds_c(u) \leq lb - 2$, then $\theta(G,k)=\theta(G[V\setminus \{u\}])$. 
\end{lemma}

The correctness of Lemma 4 can also be directly derived from Lemma 1, using the fact that \( \omega(G[N(u)]) \leq ds_c(u) \).



Given a coloring $c$ of the graph $G$, suppose \( \theta(G, k) \ge lb \). 
We can exhaustively remove all vertices \( u \in V \) that \( ds_c(u) \leq lb - 2 \) based on this reduction rule. 
The time complexity of this operation can be simply achieved in linear time $O(m)$.
In order to obtain an initial coloring, we use the heuristic coloring algorithm in the state-of-the-art maximum clique algorithm~\cite{li2017minimization}.
Indeed, we can remove more vertices if the colors of the vertices are dynamically adjusted during the reduction process.
This results in a more complex reduction process, with a running time of $O(nm)$.
The details are provided in the appended file.


\paragraph{Triangle Reduction}
For an edge \( \{u, v\} \), if the number of common neighbors of \( u \) and \( v \), $|N(u)\cap N(v)|$, is less than \( \theta(G, k)-2 \), then the size of the clique containing both vertices is smaller than \( \theta(G, k)-2 \). 
This implies the triangle reduction rule.


\begin{lemma}[Triangle Reduction]
Given a graph $G=(V,E)$, an integer $k$ and an integer $lb$, if $ \theta(G,k) \geq lb$ and there is an edge $\{u,v\}\in E$ such that $|N(u)\cap N(v)|  \leq lb - 3$  then $\theta(G,k)=\theta((V,E\setminus \{\{u,v\}\}),k)$. 
\end{lemma}



Suppose that \( \theta(G, k) \ge lb \). The triangle reduction rule indicates that we can exhaustively remove all edges \( \{u,v\} \in E \) that \( |N(u)\cap N(v)| \leq lb - 3\) based on this reduction rule. 
In order to identify such edges, we can list all triangles (3-cliques) in the $G$, which can be done in time \( O(m^{1.5}) \) \cite{latapy2008main}.
Then, the size of $|N(u)\cap N(v)|$ is equal to the number of triangles including edge $\{u,v\}$.
Moreover, if $\{w,u,v\}$ forms a triangle, the removal of edge $\{u,v\}$ also makes edges $\{w,u\}$ or $\{w,v\}$ involved in fewer triangles. 
In other words, we need to keep updating the number of triangles that each edge is involved with. 
This can be also done in linear time $O(3m)$ because each edge can be removed at most once.
In total, the running time of exhaustively removing all edges is $O(m^{1.5})$.

The relationship between triangle reduction and edges is analogous to that between degree reduction and vertices. 
Similarly, we could extend this concept to reductions based on color numbers or maximum cliques of common neighbors.
\begin{lemma}[Triangle Clique Reduction]
Given a graph $G=(V,E)$, an integer $k$ and an integer $lb$, if $ \theta(G,k) \geq lb$ and there is an edge $\{u,v\}\in E$ such that $\omega(G[N(u)\cap N(v)])  \leq  \theta(G,k)-3$, then $\theta(G,k)=\theta((V,E\setminus \{\{u,v\}\}),k)$. 
\end{lemma}
\begin{lemma}[Triangle Color Reduction]
Given a graph \( G = (V, E) \), an integer \( k \), and an integer $lb$, if $\theta(G,k)\geq lb$ and there is an edge $\{u,v\}\in E$ such that \( |\{c(w):w\in N(u)\cap N(v)\}|\leq\theta(G, k)-3 \), then \( \theta(G, k) = \theta\left( (V, E \setminus \{\{u, v\}\}), k \right) \).
\end{lemma}

Clearly,  exhaustively removing edges that satisfy the triangle clique or triangle color reduction rules is more computationally expensive than that of triangle reduction. 
On the other hand, after exhaustively applying the degree, color, and triangle reductions, the graph becomes significantly denser, making it hard to remove edges meeting the triangle clique and color reduction rule.
Thus, the above two reductions are only used when a long computational time is allowed.


\paragraph{Interdiction Reduction}
The previous reductions primarily identify vertices that are not in $\mathcal{S}(G,k)$ or not in any maximum clique of $G$. 
In contrast, now we introduce interdiction reduction rules that identify vertices that must be a part of every maximum clique in the remaining graph.


\begin{lemma}[Interdiction Reduction]
Given a graph \( G = (V, E) \), an integer \( k > 0\), and a vertex \( u \in V \), if for any $v \in V\setminus N[u]$, \( \omega(G[N(u)]) - k > \omega(G[N(v)]) \), then \( \theta(G, k) = \theta(G[V \setminus \{u\}], k - 1) \).  
\end{lemma}


In other words, if the largest clique containing vertex $u$ is larger than any maximum clique that excludes $u$ by more than $k$, then in any solution where at most $k$ vertices are removed and $u$ is not among them, the largest remaining clique must include $u$.



We can remove all vertices \( u \in V \) that for any $v \in V\setminus N[u]$, \( \omega(G[N(u)]) - k > \omega(G[N(v)]) \) based on this reduction rule. 
The time complexity of this reduction is \( O\left( n^2 \right) \) when we have already computed \( \omega(G[N(u)]) \) for all vertices \( u\in V \) in the exact clique reduction above.



\paragraph{Domination Reduction}
If all the neighbors of a vertex \( v \) are also neighbors of another vertex \( u \), and \( u \) has additional neighbors that are not connected to \( v \), then removing vertex \( u \) is always more effective than removing vertex \( v \). We call this property \( u \) dominates \( v \).

\begin{lemma}[Domination Reduction]
\label{domred}
Given a graph \( G = (V, E) \), an integer \( k \), and two vertices \( u, v \in V \), 
if \( N(v) \subset N(u) \) or \( N[v] \subset N[u] \), then if there exists a set \( S \in \mathcal{S}(G, k) \) such that \( v \in S \) and \( u \notin S \), then \( (S \setminus \{v\}) \cup \{u\} \in \mathcal{S}(G, k) \).  
\end{lemma}

We can find all pairs of vertices \( u,v \in V \) that $u$ dominates $v$ based on this reduction rule. 
The time complexity of this reduction is \( O\left( nm \right) \) because the time complexity of determining the domination relationship for a pair of vertices \( (u, v) \) is \( O(d(u) + d(v)) \).




\section{The Reduction-and-Search Algorithm}

\subsection{Lower Bound Estimation}
\label{subsec-lowerbound}
The proposed reduction rules rely on a known lower bound of $\theta(G,k)$. 
In this section, we elaborate on the algorithm to obtain an effective lower bound.

\paragraph{The Linear Program Relaxation}
Given a graph \( G = (V, E) \) and an integer \( k \), it is known that \( \theta(G, k) \) can be computed using the following ILP formulation \cite{furini2019maximum}.
\[
\text{ILP-CIP}(G, k, \mathcal{C})\left\{
\begin{aligned}
\min \quad & y \\
\text{s.t.} 
& \sum_{u\in V} x_u \geq n - k, \\
& \sum_{u \in C} x_u \leq y, \quad \forall C \in \mathcal{C},\\
&y\in \mathbb{R}, x_u \in\{0,1\}, \quad \forall u \in V 
\end{aligned}
\right.
\]
where \( \mathcal{C} \) is the set of all cliques in the graph \( G \). Note that the binary $x_u=0$ means that vertex $u$ is interdicted.

Clearly, we can obtain a lower bound for \( \theta(G, k) \) by relaxing the ILP.
A common relaxation technique is modifying the domain of \( x_u \) to all real values in the range \([0, 1]\). However, it is still difficult to obtain the optimal solution of this LP due to the exponential number of cliques in $G$. Yet it is hard to separate the inequality in polynomial time.

Due to the exponential size of \( |\mathcal{C}| \), in the paper, we provide an alternative relaxation by restricting \( \mathcal{C} \) to only a subset of the maximal cliques in the graph. We note that the following observation holds.

\begin{lemma}
\label{lemma_lowerbound}
Assume $\mathcal{C}= \{C_1,...,C_p\}$ where $C_1,...,C_p$ are cliques in $G$ (,  probably not all cliques in $G$). Then the optimal solution to ILP-CIP$(G,k,\mathcal{C})$ is a lower bound of $\theta(G,k)$.
\end{lemma}



\paragraph{Restrict $\mathcal{C}$ to Disjoint Cliques}
By Lemma \ref{lemma_lowerbound}, when \( \mathcal{C} \) only includes some mutually disjoint cliques of $G$, the ILP-CIP$(G, k, \mathcal{C})$ is clearly a lower bound of $\theta(G,k)$. We refer to this lower bound as the \textit{disjoint lower bound}.

To compute a disjoint lower bound, we first construct a set \( \mathcal{C} \) which consists of disjoint cliques in $G$ using the simple Algorithm 2 in the appended file. This algorithm uses the simple greedy strategy and has a time complexity of $O(n^2)$.

If $\mathcal{C}$ consists solely of disjoint cliques of $G$, then ILP-CIP$(G, k, \mathcal{C})$ can be computed without the need to invoke an ILP solver.
Denote \( f(\mathcal{C}, y) \) as the minimum number of vertices that need to be removed such that the size of the largest clique in \( \mathcal{C} \) is less than or equal to \( y \), i.e., \(f(\mathcal{C}, y) = \sum_{C \in \mathcal{C}} \max(0, |C| - y)\).
The smallest \( y \) such that \( f(\mathcal{C}, y) \leq k \) is the result of ILP-CIP$(G, k, \mathcal{C})$. Because \( f(\mathcal{C}, y) \) is non-decreasing as \( y \) increases, allowing us to use binary search to determine \( y \),  ILP-CIP$(G, k, \mathcal{C})$ is obtained in $O(|\mathcal{C}| \log n)$. 
In sum, the time to obtain a disjoint lower bound is $O(n^2+|\mathcal{C}| \log n)=O(n^2)$ when $\mathcal{C}$ only contains disjoint cliques.

\paragraph{Relaxation to Bipartite Cliques}
Furthermore, supposing  that \( \mathcal{C}\) can be partitioned into $\mathcal{C}_1$ and $\mathcal{C}_2$, where \( \mathcal{C}_1 \) includes some mutually disjoint cliques of $G$, and  \( \mathcal{C}_2 \) also includes some mutually disjoint cliques of $G$,
then ILP-CIP$(G, k, \mathcal{C})$ is also a lower bound of $\theta(G,k)$ by Lemma \ref{lemma_lowerbound}.  We refer to this lower bound as the \textit{bipartite lower bound}.

In order to obtain $\mathcal{C}_1$ and $\mathcal{C}_2$, we simply run the greedy Algorithm 2 in the appended file twice. 
Then, we use a flow-based algorithm to compute ILP-CIP$(G, k, \mathcal{C}=\mathcal{C}_1\cup \mathcal{C}_2)$.

For notational convenience, let us denote \( f'(\mathcal{C}=\mathcal{C}_1\cup \mathcal{C}_2, y) \) as the minimum number of vertices that need to be removed such that the size of the largest clique in \( \mathcal{C} \) is less than or equal to \( y \), i.e. \(f'(\mathcal{C}, y) = \sum_{C \in \mathcal{C}_1} \max(0, |C| - y) + \sum_{C \in \mathcal{C}_2} \max(0, |C| - y) - h(\mathcal{C}_1, \mathcal{C}_2, y)\). Here, \(h(\mathcal{C}_1, \mathcal{C}_2, y) \) denotes the maximum number of vertices that can be removed to simultaneously reduce the size of cliques in \( \mathcal{C}_1 \) and \( \mathcal{C}_2 \) to at least \( y \).

\begin{figure}[h]
    \centering
    \includegraphics[width=0.85\linewidth]{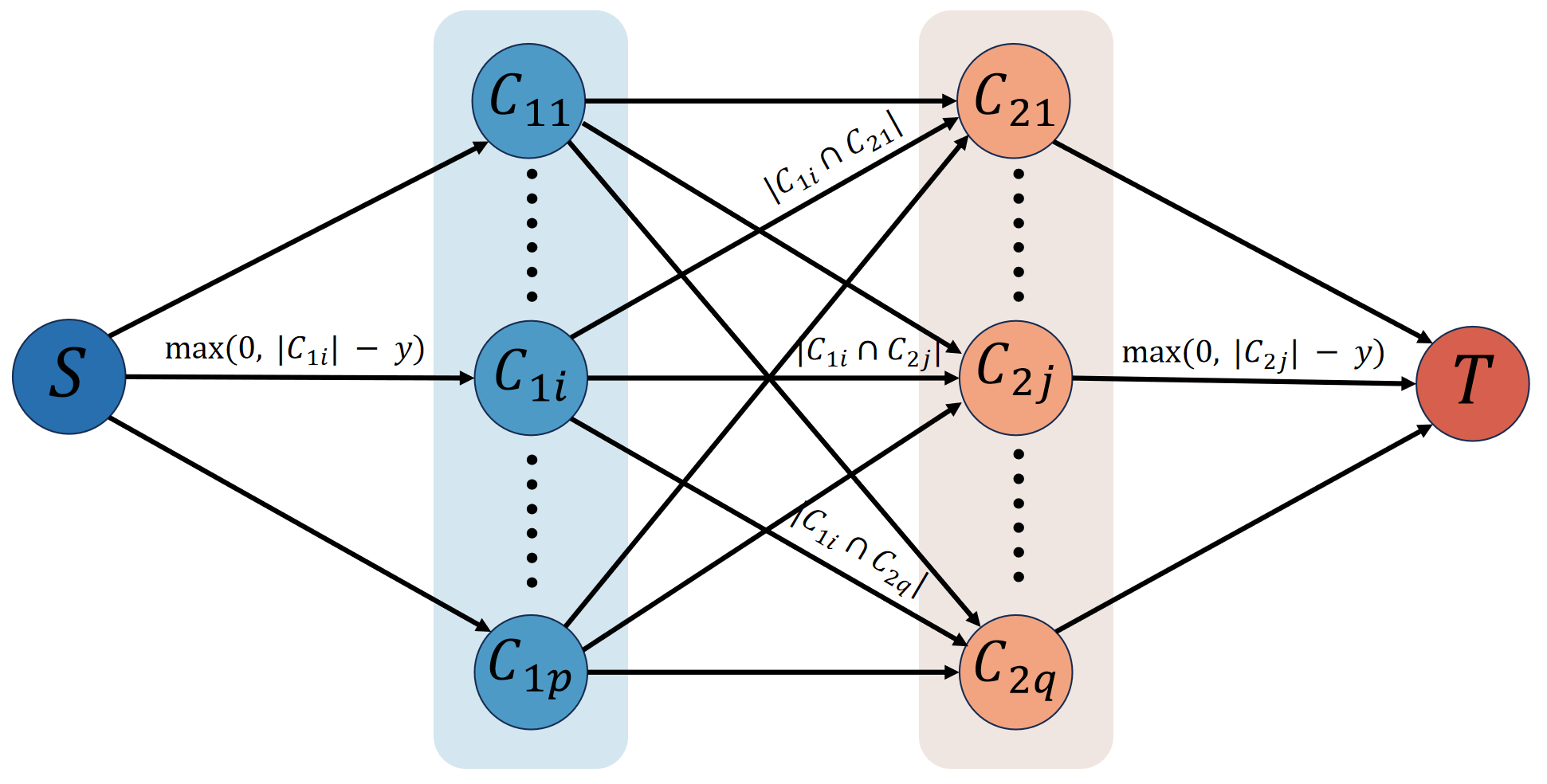}
    \caption{Structure of the network flow model, the value on the side represents the flow size.}
    \label{fig:wll}
\end{figure}

Given sets $\mathcal{C}_1$ and $\mathcal{C}_2$, and an integer $y$, the \(h(\mathcal{C}_1, \mathcal{C}_2, y) \) is computed using a maximum flow  algorithm.
We first build a flow network as shown in Figure~\ref{fig:wll}. 
For each clique $C$ in \( \mathcal{C}_1 \) and \( \mathcal{C}_2 \), we create a node in the network representing the $C$. 
We also create an additional source node $S$ and a target node $T$.
Each node representing a \( C_1 \in \mathcal{C}_1 \) connects to the source $S$ with capacity \( \max(0, |C_1| - y) \), and each node representing \( C_2 \in \mathcal{C}_2 \) connects to the target $T$ with capacity \( \max(0, |C_2| - y) \). 
Edges between \( C_1\in \mathcal{C}_1 \) and \( C_2\in\mathcal{C}_2 \) have capacity \( |C_1 \cap C_2| \). 
The maximum $(S,T)$-flow in this flow network is equal to \(h(\mathcal{C}_1, \mathcal{C}_2, y) \). 
The total time of building the flow network and computing the maximum flow is \( O(n^3) \).

As \(h(\mathcal{C}_1, \mathcal{C}_2, y) \) can be computed efficiently,  it is easy to compute ILP-CIP$(G, k, \mathcal{C}=\mathcal{C}_1\cup \mathcal{C}_2)$.
Similar to the disjoint lower bound, the smallest \( y \) such that \( f'(\mathcal{C}=\mathcal{C}_1\cup \mathcal{C}_2, y) \leq k \) is the result of ILP-CIP$(G, k, \mathcal{C}=\mathcal{C}_1\cup \mathcal{C}_2)$. 
Thereby, the total running time of computing the bipartite lower bound is \( O(n^3\log n )\).

\subsection{The Reduction-based Preprocess}
\begin{figure}[h]
    \centering
    \includegraphics[width=0.9\linewidth]{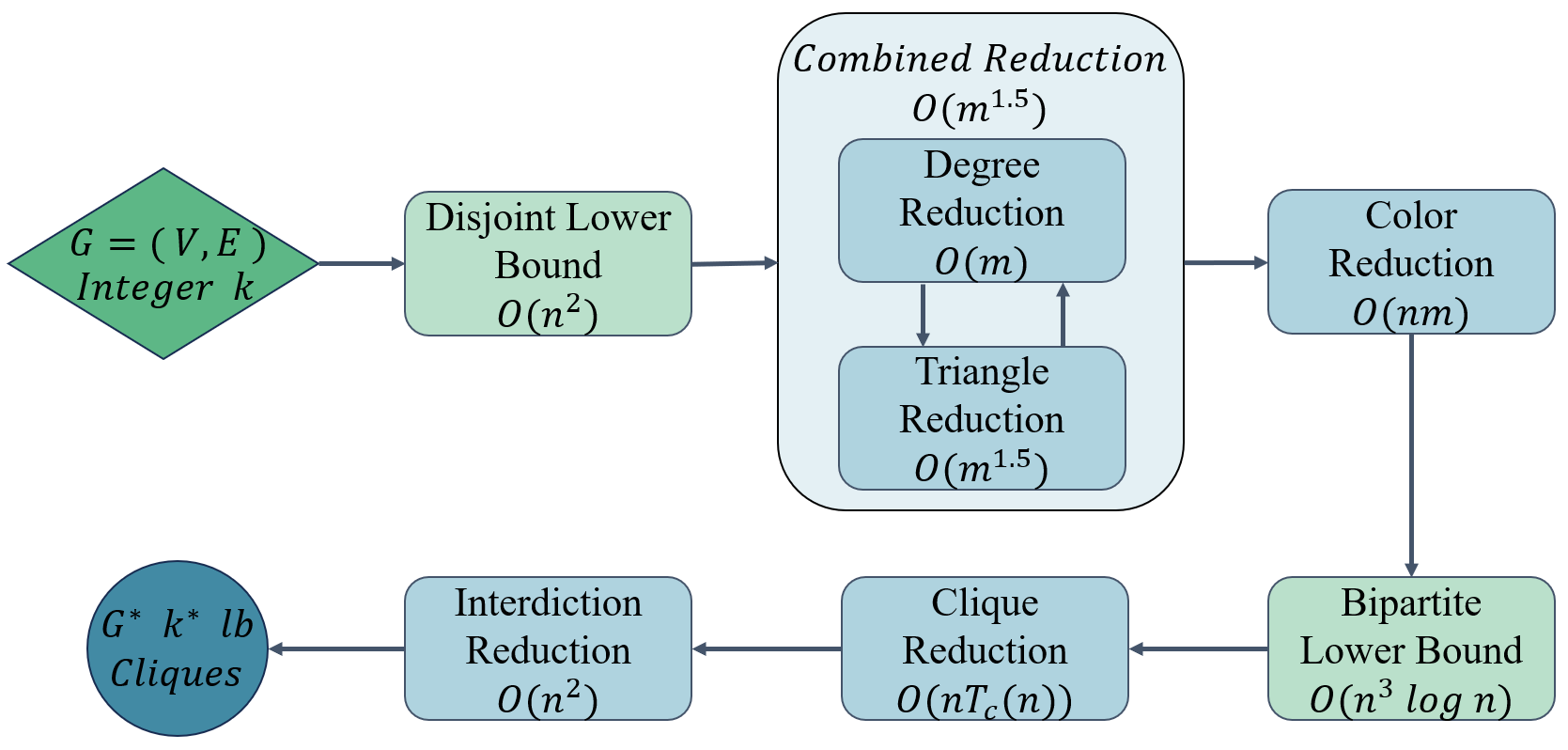}
    \caption{Flowchart of the reduction-based preprocess,  time complexities are given to each step.}
    \label{fig:lct}
\end{figure}

Firstly, these reduction rules require a valid lower bound of $\theta(G,k)$.
Hence, the initial step is to use the simple disjoint cliques relaxation method to obtain a lower bound.

Afterwards, we call the degree reduction and triangle reduction to reduce the $(G,k)$. 
The two reductions can be done in an intertwined manner. 
Specifically, the deletion of an edge decreases the degree of its connected vertices, while the deletion of a vertex affects the triangles that an edge involves.
Hence, we maintain the vertex degrees and triangles simultaneously and update them when a vertex (and its incident edges) or an edge is removed from the graph.
The total time complexity of this combined reduction step remains \( O(m^{1.5}) \).

When the graph can no longer be reduced by degree or triangle reduction, we apply color reduction, followed by clique reduction. 
It is known that the computation of clique reduction is relatively expensive.
To achieve a better trade-off between reduction time and effectiveness, we use the bipartite lower bound before the clique reduction.
Intuitively, a tighter lower bound can help reduce more vertices in the clique reduction. 

So far, the reductions are used in an increasing order of their complexity. 
Exceptions are given to the interdiction reduction.
The interdiction reduction step follows the clique reduction because it asks for  $\omega(G[N(u)])$ for each $u\in V$.
Finally, the pre-processing outputs graph $G$ and budget $k$.




\subsection{Integrating the Reduction and Branch-and-cut}
\begin{algorithm}[h]
    \caption{RECIP}
    \label{alg:algorithm}
    \textbf{Input}: Graph $G=(V,E)$, an integer $k$\\
    \textbf{Output}: $\theta(G,k)$
    \begin{algorithmic}[1] 
        \STATE $G,k,lb,Cliques\gets preprocess(G,k)$
        \STATE $ans \gets branc\_and\_cut(G,k,lb,Cliques)$
        \STATE \textbf{return} $ans$
    \end{algorithmic}
\end{algorithm}

Our algorithm, RECIP, consists of two main components. The first component is the reduction-based preprocessing step described earlier, which aims to reduce the size of the graph. The second component is a branch-and-cut algorithm \cite{mitchell2002branch} to solve the ILP-CIP formulation.

\paragraph{Initial Linear Program} 

We use ILP-CIP$(G,k,\mathcal{C})$ to build the linear program, where $\mathcal{C}$ are cliques obtained by the reduction process.
Additionally, based on the domination reduction rule (Lemma~\ref{domred}), we find all dominance relations between any $u$ and $v\in V$. 
Then we add $x_u\leq x_v$ when $x_u$ dominates $x_v$  to tighten the LP.

\paragraph{Speciation Oracle}
When a solution $(x^*,y^*)$ is obtained, the branch-and-cut algorithm asks a separation oracle to decide if there is a clique $C$ such that the inequality $\sum_{u\in C}x^*_u\le y$ is not satisfied. 
If so, a new inequality of the form $\sum_{u\in C}x_u\le y$ is added to continue the search.
This separation oracle is equal to the maximum clique problem in the graph \( G[V^*] \) where $V^*=\{u\in V: x^*_u=1\}$ \cite{furini2019maximum}.
Again, we use the algorithm in \cite{chang2019efficient} to find maximum cliques. 




\section{Experiments}

In this section, we evaluate the performance of our algorithm. 
As mentioned, there are two existing methods, the general BILP solver~\cite{becker2017bilevel,tang2016class,fischetti2017new} and CLIQUE-INTER \cite{furini2019maximum}, for solving the CIP.
According to the experiments in \cite{furini2019maximum}, CLIQUE-INTER is faster than the BILP solver by several orders of magnitude on all tested cases.
Therefore, we mainly compare our RECIP with CLIQUE-INTER in the experiments. 


\paragraph{Experimental Setup}
All experiments were conducted on a machine equipped with an AMD R9-7940HX CPU and 16GB of RAM, running Ubuntu 22.04. 
The codes were written in C++ and compiled with GCC version 11.4.0 using the -O2 optimization flag. 
The IBM CPLEX solver version 12.7 was used as the underlying solver for the integer linear program. 
All algorithms run in single-threaded mode with the default settings of CPLEX. 

\paragraph{Real-world Dataset}
Our experiments use 124 real-world networks sourced from the SNAP database and the Network Repository\cite{snap,networkrepository}  \footnote{The dataset can be downloaded from \url{https://lcs.ios.ac.cn/~caisw/graphs.html}.}.
The dataset covers a variety of categories, including social networks, technological networks, biological networks, and more.
These graphs can be as large as 1.9 million vertices and 5 million edges.
We leave a detailed description in the appended file.


\subsection{Results on Real-world Network Graph}


We compare the performance of the two algorithms under four different budget values (\( k \in \{\lceil 0.005|V|\rceil, \lceil 0.01|V|\rceil, \lceil 0.02|V|\rceil, \lceil 0.05|V|\rceil\} \)) with a runtime limit of 600 seconds.

\begin{figure*}[h]
    \centering
    \begin{subfigure}[b]{0.24\linewidth}
        \centering
        \includegraphics[width=\linewidth]{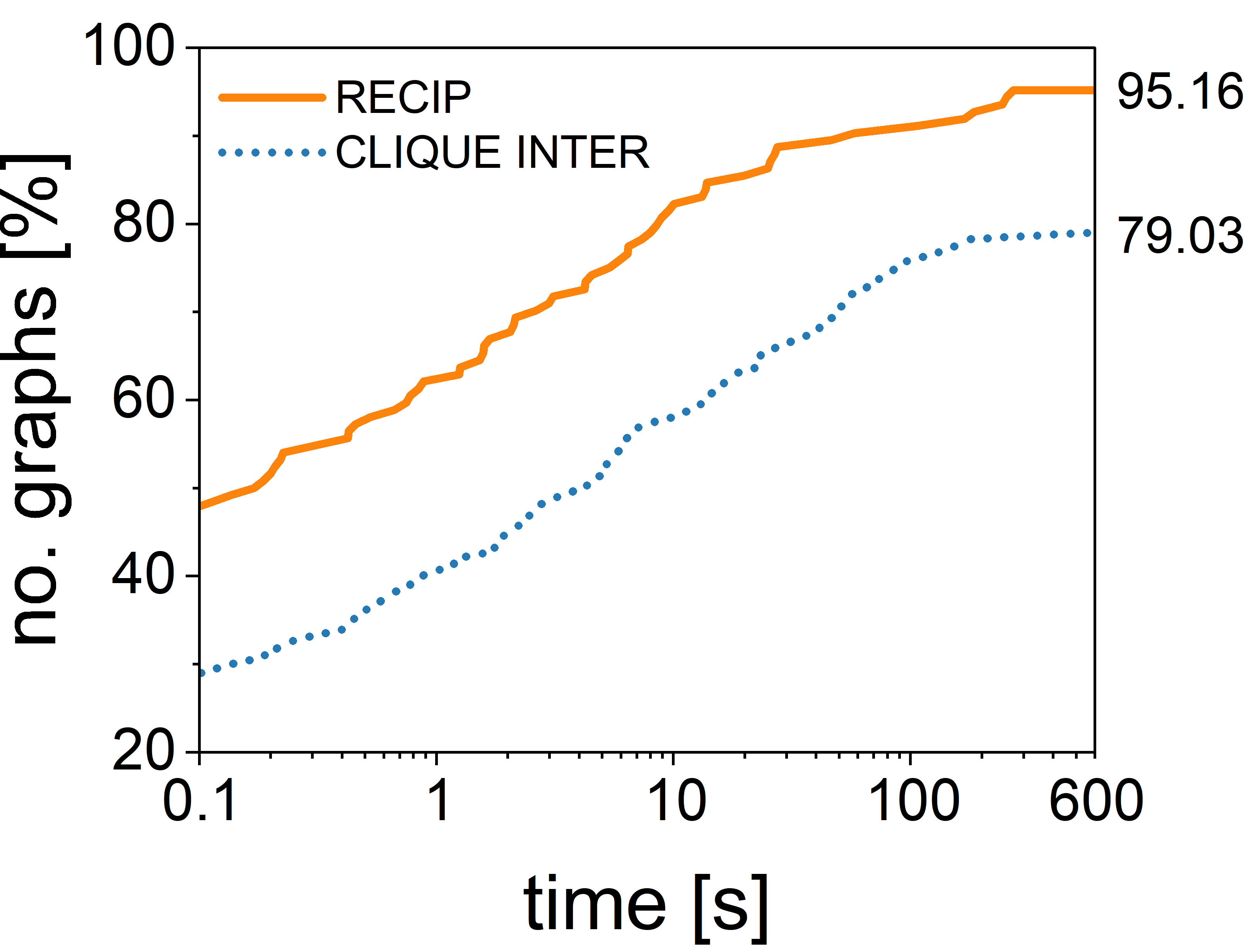}
        \caption{$k =\lceil 0.005|V|\rceil$}
        \label{fig:real-world01}
    \end{subfigure}
    \begin{subfigure}[b]{0.24\linewidth}
        \centering
        \includegraphics[width=\linewidth]{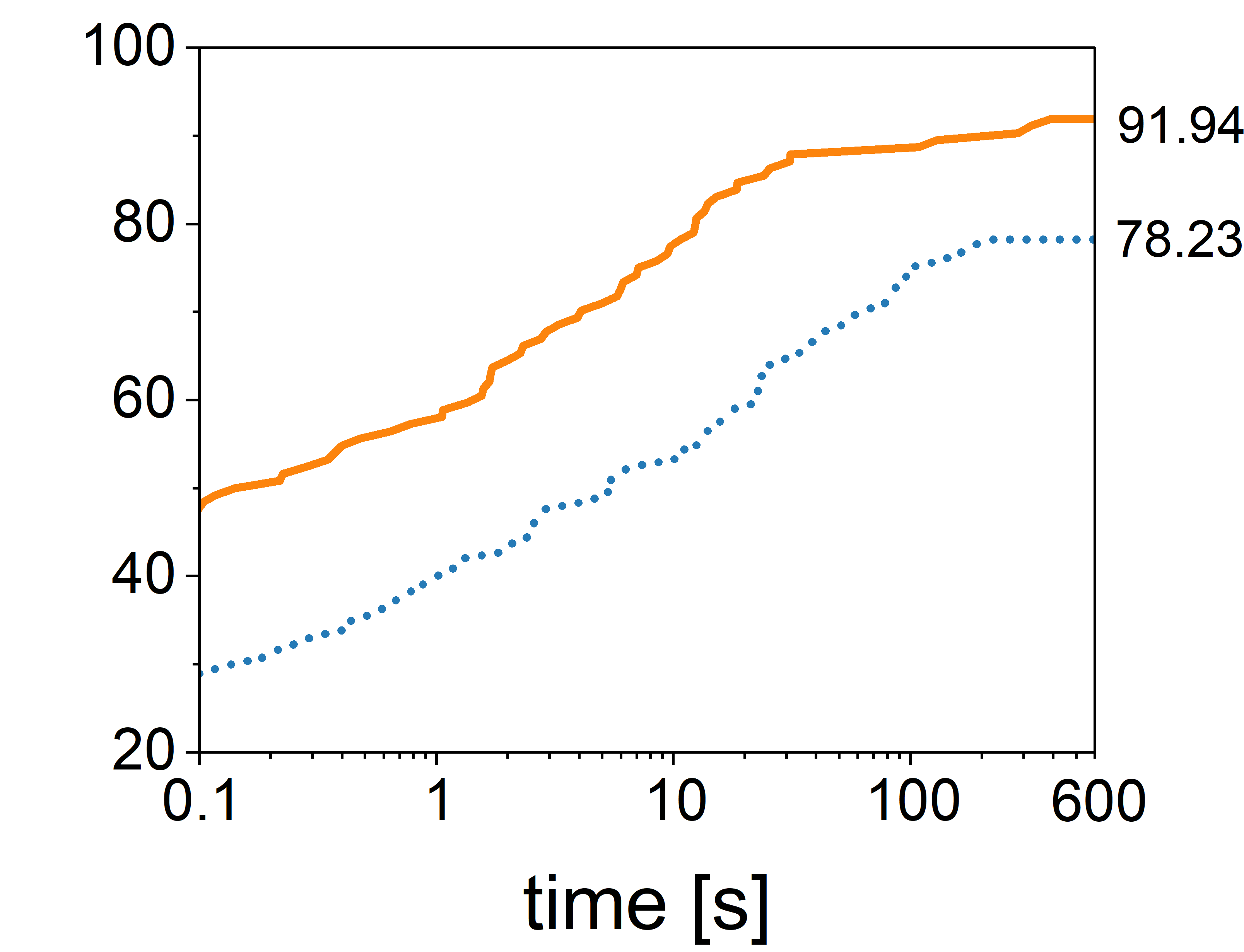}
        \caption{$k =\lceil 0.01|V|\rceil$}
        \label{fig:real-world02}
    \end{subfigure}
    \begin{subfigure}[b]{0.24\linewidth}
        \centering
        \includegraphics[width=\linewidth]{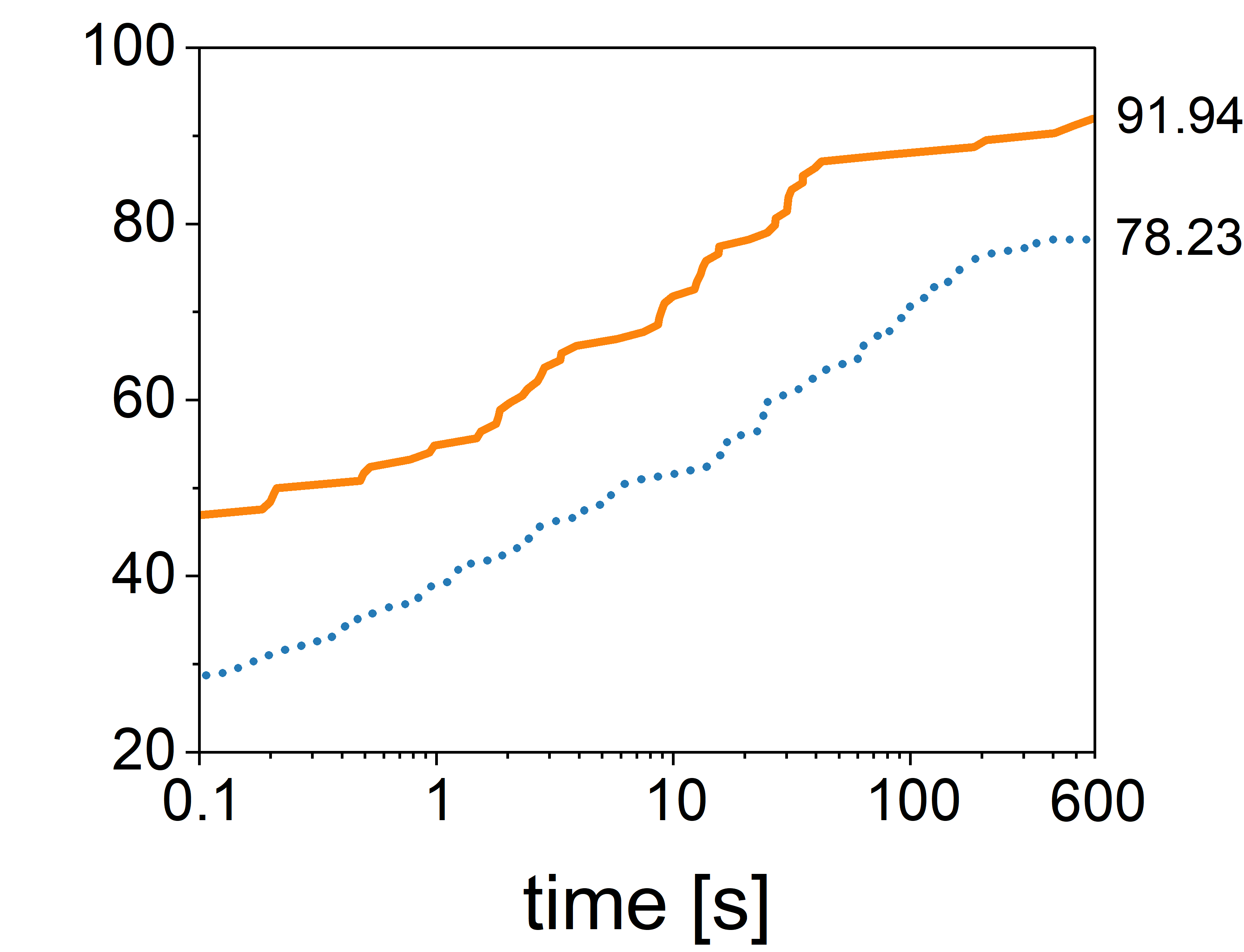}
        \caption{$k =\lceil 0.02|V|\rceil$}
        \label{fig:real-world03}
    \end{subfigure}
    \begin{subfigure}[b]{0.24\linewidth}
        \centering
        \includegraphics[width=\linewidth]{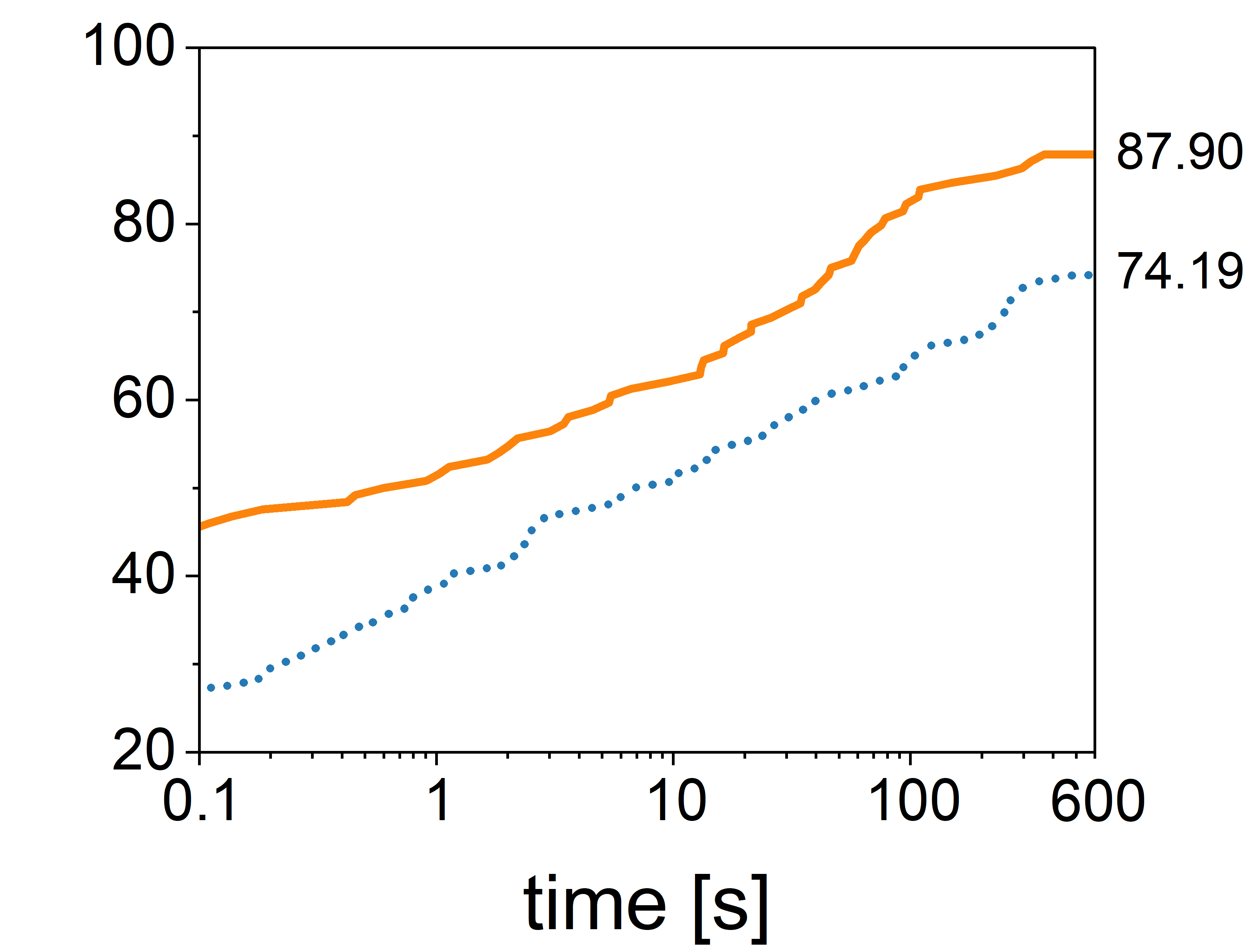}
        \caption{$k =\lceil 0.05|V|\rceil$}
        \label{fig:real-world04}
    \end{subfigure}
    \caption{The proportion of instances solved by both algorithms within the time range of 0.1 to 600 seconds on real-world network graphs under different \( k \) values.}
    \label{fig:real-world}
\end{figure*}

Figure~\ref{fig:real-world} shows the number of instances solved within different time limits for each budget value \( k \). 
For every $k$, RECIP consistently outperforms CLIQUE-INTER across all time frames, solving over 10\% more instances after 600s. 
Furthermore, for a fixed time frame, the number of instances solved by both algorithms decreases as \( k \) increases.
This trend matches the observation that fewer vertices and edges can be reduced when $k$ increases.

\begin{figure}[h]
    \centering
    \includegraphics[width=0.85\linewidth]{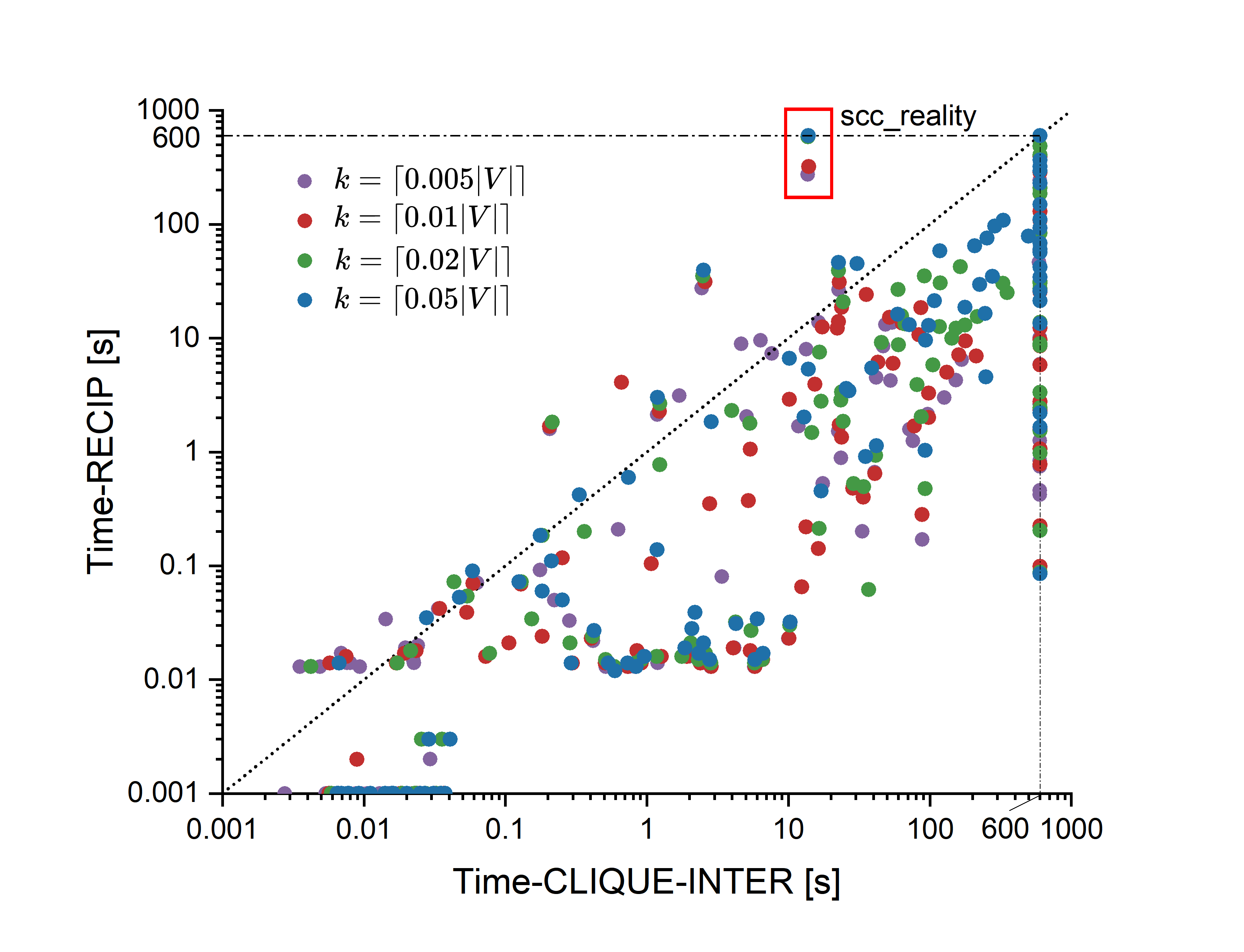}
    \caption{Runtime of both algorithms for each instance, with instances exceeding the 600-second time limit plotted at the boundary. For clarity, a dashed line representing \( x = y \) is added.}
    \label{fig:time-time}
\end{figure}

In the scatter plot in Figure~\ref{fig:time-time},  each point represents an instance, i.e., a pair $(G,k)$, where the axes represent the runtime of the two algorithms. 
Points below the diagonal line indicate cases where RECIP is faster than CLIQUE-INTER, while points above indicate the opposite.  
There are 410 instances that RECIP outperforms CLIQUE-INTER, accounting for 89.7\% of the total instances. 
We note that there are 39 instances that both algorithms cannot solve within the time limit.
Furthermore, we observed that in some relatively denser graphs, the reduction process can be highly time-consuming. For example, in graph \texttt{scc\_reality}, where \(|V| = 6809\) and \(|E| = 4714485\), when \( k = \lceil 0.05n \rceil \), the clique reduction step takes 816 seconds but only removes 5 vertices. 
In contrast, the subsequent branch-and-cut process completes in just 1.6 seconds. 

\subsection{Effectiveness of Reduction Rules}



Now, we give a detailed break-up analysis of the reduction algorithms.
For each instance, we record the graph size after each reduction step, i.e., degree, triangle, color, and exact clique reductions, and we also record the size of the remaining graph.
Due to space limitations, we report the 10 graphs with the most vertices from those that can be solved by RECIP within 600s and where not all vertices are removed during the reduction process in Figure \ref{fig:red}. 
Clearly, the simple degree reduction removes at least half of the vertices for the majority of graphs. 
This is within our expectation as these graphs are sparse.
In contrast, the interdiction reduction step removes few vertices.
Nevertheless, the remaining reduction steps still play an important role in reducing around 20\% of the number of vertices, especially when \( k = \lceil 0.05|V| \rceil \).
Increasing the value of \( k \) makes the reduction process less effective, resulting in a greater number of remaining vertices. 

\begin{figure*}[t]
    \centering
    \begin{subfigure}[b]{0.22\linewidth}
        \centering
        \includegraphics[angle=90, width=\linewidth]{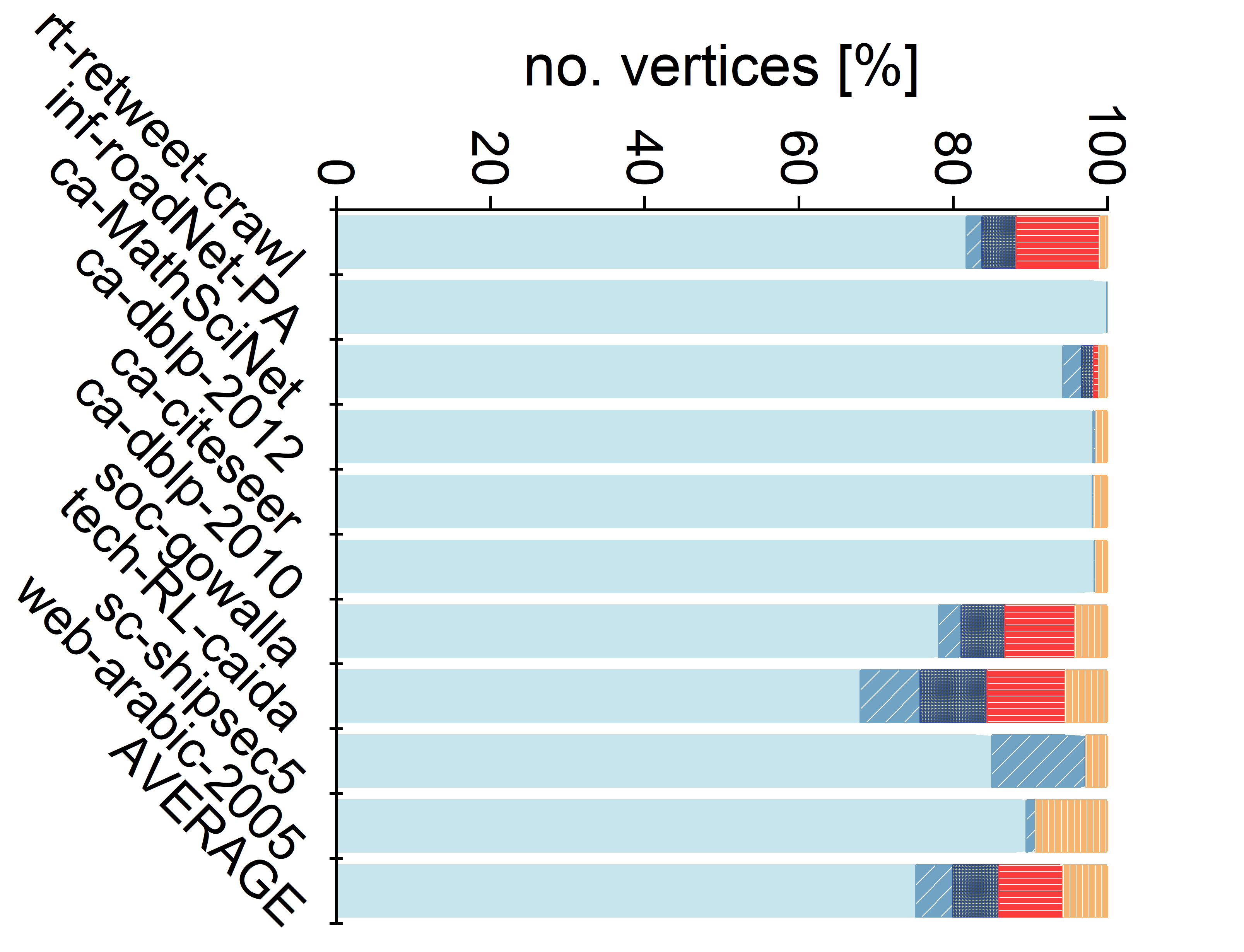}
        \caption{$k =\lceil 0.005|V|\rceil$}
        \label{fig:real-world1}
    \end{subfigure}
    \vspace{0.1cm}
    \begin{subfigure}[b]{0.22\linewidth}
        \centering
        \includegraphics[angle=90,width=\linewidth]{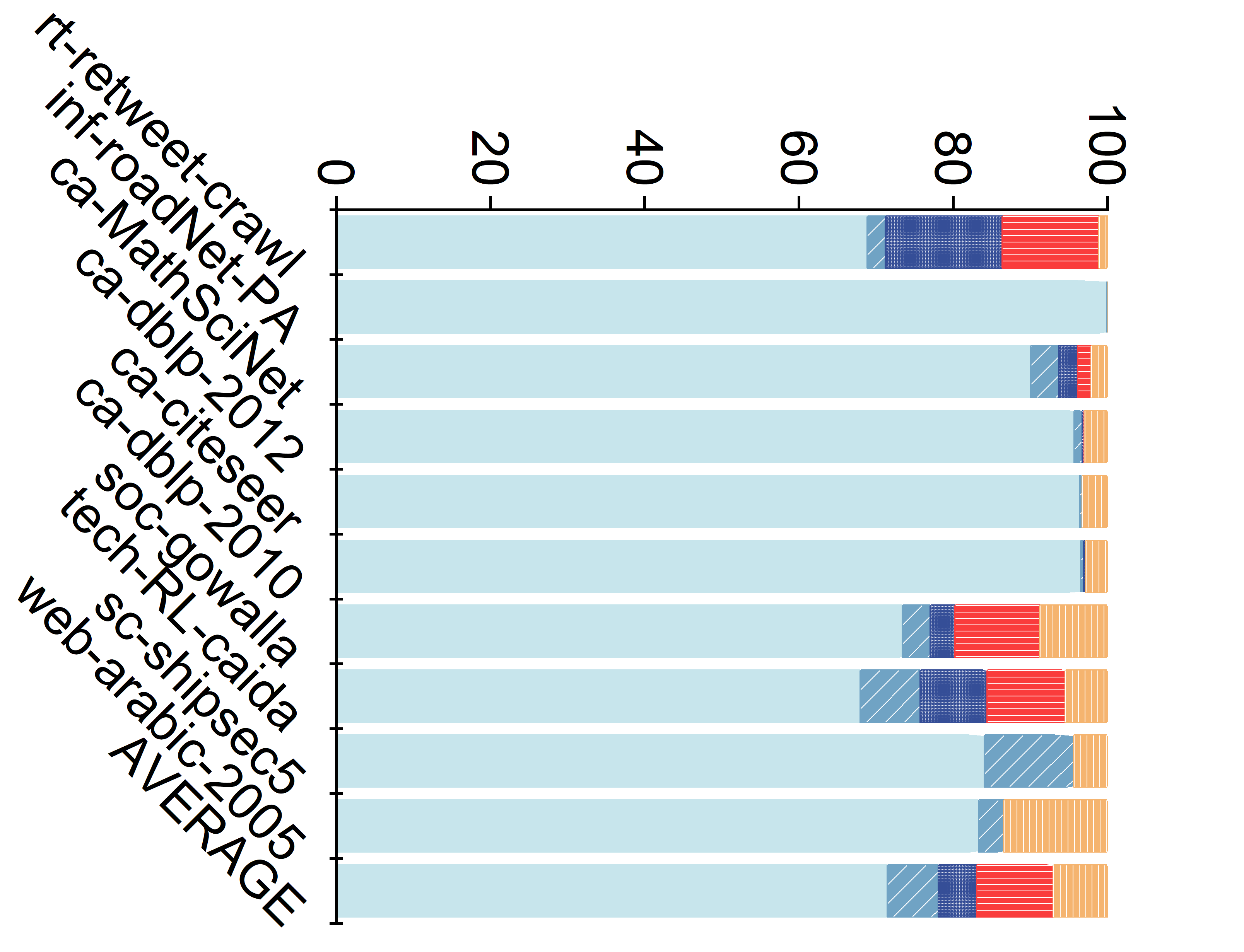}
        \caption{$k =\lceil 0.01|V|\rceil$}
        \label{fig:real-world2}
    \end{subfigure}
    \vspace{0.1cm}
    \begin{subfigure}[b]{0.22\linewidth}
        \centering
        \includegraphics[angle=90,width=\linewidth]{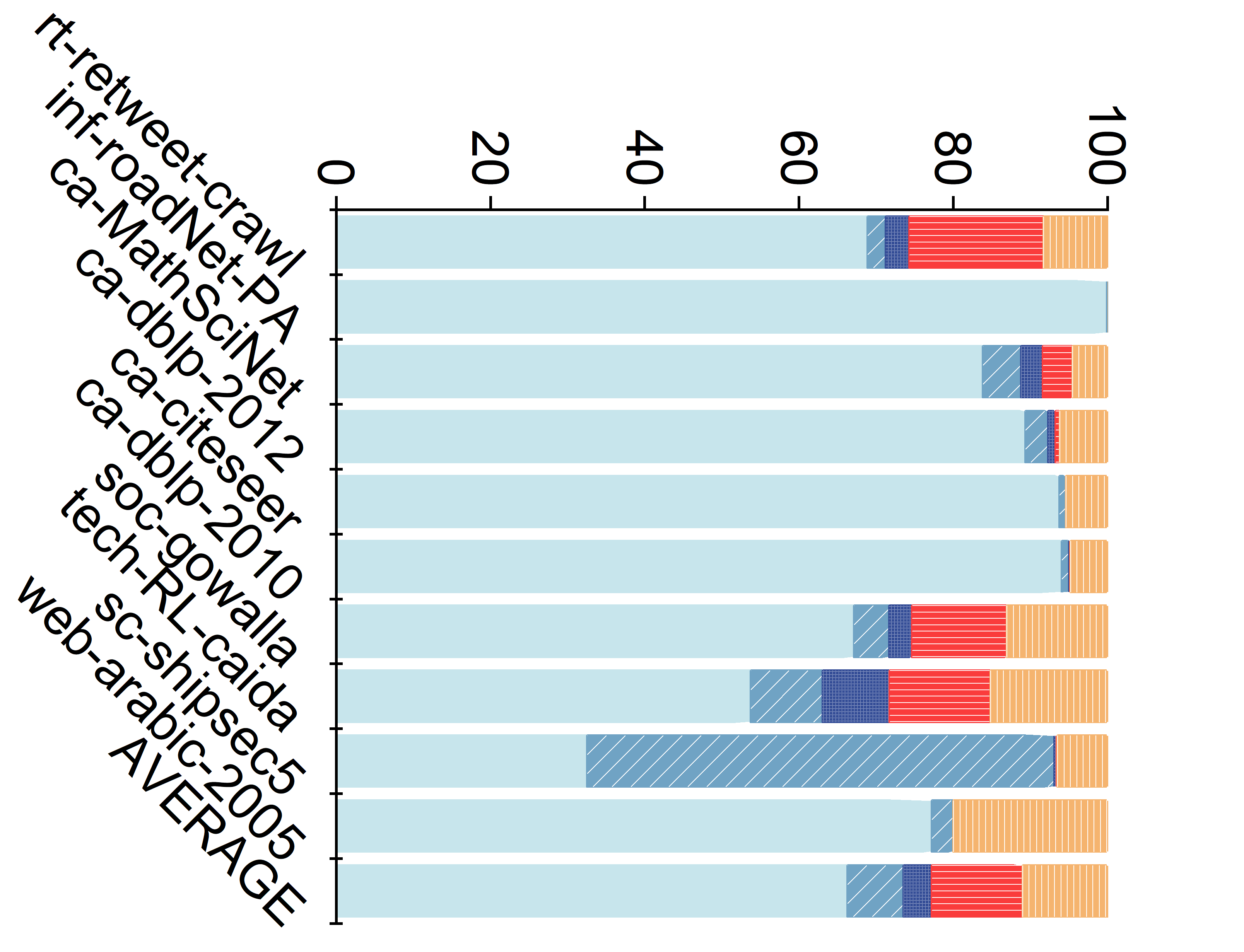}
        \caption{$k =\lceil 0.02|V|\rceil$}
        \label{fig:real-world3}
    \end{subfigure}
    \vspace{0.1cm}
    \begin{subfigure}[b]{0.22\linewidth}
        \centering
        \includegraphics[angle=90, width=\linewidth]{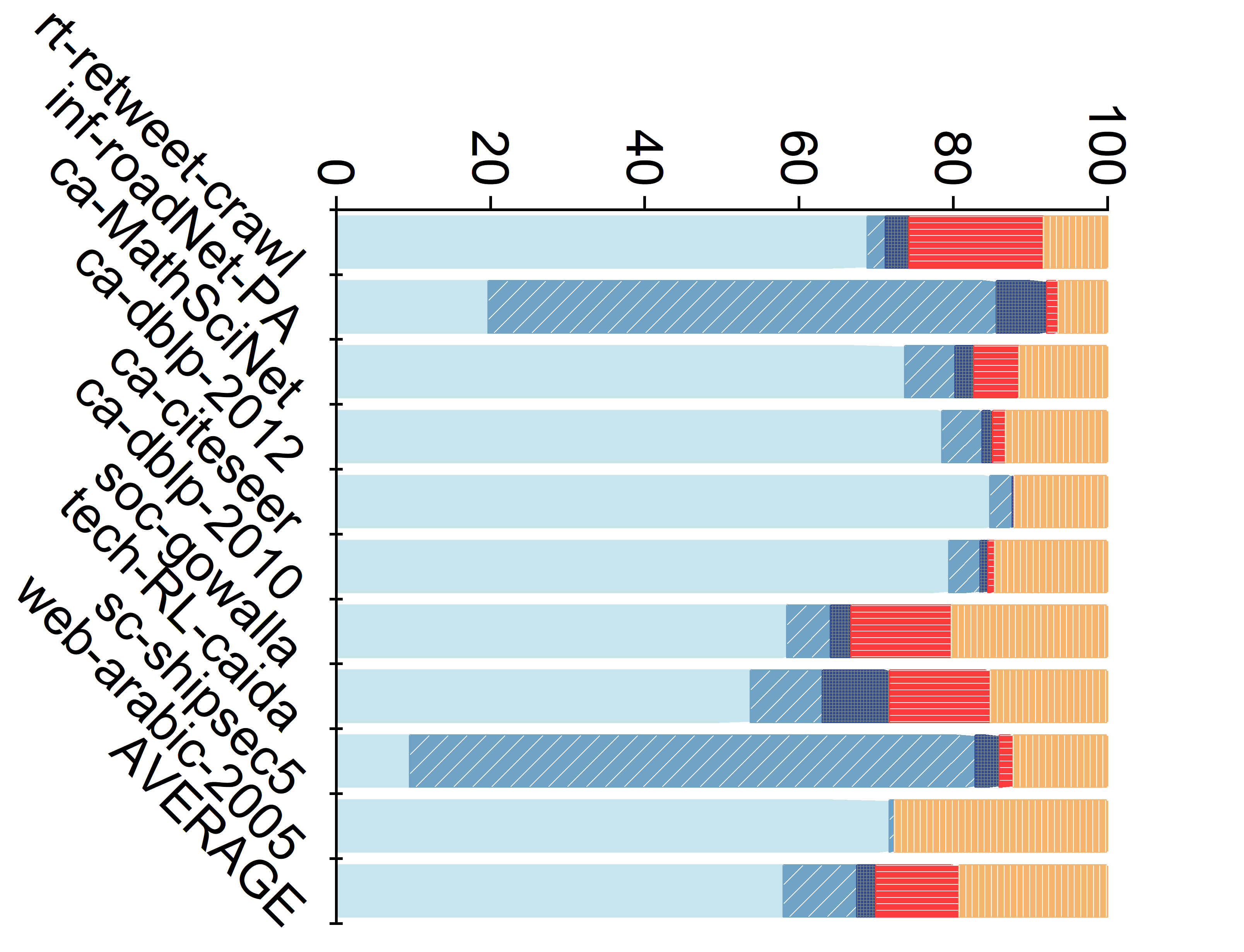}
        \caption{$k =\lceil 0.05|V|\rceil$}
        \label{fig:real-world4}
    \end{subfigure}
    \vspace{0.1cm}
    \begin{subfigure}[b]{0.44\linewidth}
        \centering
        \includegraphics[width=\linewidth]{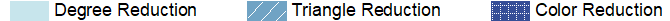}
        \label{fig:real-world5}
    \end{subfigure}
    \begin{subfigure}[b]{0.44\linewidth}
        \centering
        \includegraphics[width=\linewidth]{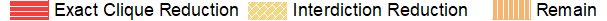}
        \label{fig:real-world6}
    \end{subfigure}
    
    \caption{Reduction proportion at each step for 10 selected graphs and the average, with different \( k \) values.}
    \label{fig:red}
\end{figure*}

\subsection{Analysis with Lower Bounds}

To demonstrate the effectiveness of our lower bound algorithm, we record the disjoint lower bound and bipartite lower bound for the 10 graphs. We present the results for \( k = 0.005|V| \) in  Table~\ref{tab:lb} and leave the rest in the appended files.
\begin{table}[htbp]
  \centering
  \caption{Disjoint and bipartite lower bound for 10 selected graphs with $k =\lceil 0.005|V|\rceil$.}
  \begin{tabular}{lrrrrr}
    \toprule
          & \multicolumn{2}{c}{disjoint} & \multicolumn{2}{c}{bipartite} & \multicolumn{1}{c}{$\theta(G,k)$} \\
    \cmidrule(r){2-3} \cmidrule(l){4-5}
          & \multicolumn{1}{c}{lb} & \multicolumn{1}{c}{time[s]} & \multicolumn{1}{c}{lb} & \multicolumn{1}{c}{time[s]} &  \\
    \midrule
    rt-retweet-crawl & 3     & 0.31  & 3     & 238.10 & 3      \\
    inf-roadNet-PA   & 3     & 0.134 & 3     & 0.201   & 3      \\
    ca-MathSciNet    & 7     & 0.067 & 7     & 2.54    & 7      \\
    ca-dblp-2012     & 15    & 0.105 & 15    & 0.506   & 15     \\
    ca-citeseer      & 25    & 0.044 & 25    & 0.317   & 25     \\
    ca-dblp-2010     & 19    & 0.039 & 19    & 0.245   & 19     \\
    soc-gowalla      & 7     & 0.056 & 8     & 23.913  & 8      \\
    tech-RL-caida    & 4     & 0.043 & 4     & 20.212  & 5      \\
    sc-shipsec5      & 21    & 0.084 & 21    & 1.005   & 21     \\
    web-arabic-2005  & 49    & 0.051 & 49    & 4.711   & 49     \\
    \bottomrule
  \end{tabular}
  \label{tab:lb}
\end{table}


We observe that the gap between our lower bound algorithm and \( \theta(G, k) \) is very small. The bipartite lower bound obtains a tighter bound for soc-gowalla with a sacrifice of running time.
Across all the 455 instances where $\theta(G,k)$ is known, the sum of the gaps between $\theta(G,k)$ and the bipartite lower bound is only 24. 
This demonstrates the effectiveness of our lower-bound algorithms.

\subsection{Analysis with Random Graphs}  

To further investigate the scalability of the algorithm, we create 5 groups of Erdős-Rényi random \( G(n, p) \) graphs, each group has a vertex number $n$($|V|$) selected from $\{50, 75, 100, 125, 150\}$.
In each group, we generate 11 graphs of edge densities (\( p \in \{0.1, 0.2, 0.3, 0.4, 0.5, 0.6, 0.7, 0.8, 0.95, 0.98\} \))\footnote{The edge density $p$ is $|E|/\binom{|V|}{2}$}, then we evaluate four different budget values (\( k \in \{\lceil 0.05|V|\rceil, \lceil 0.1|V|\rceil, \lceil 0.2|V|\rceil, \lceil 0.4|V|\rceil\} \)) for each of these graphs.
The cut-off time is still 600 seconds.

\begin{figure}[htbp]
    \begin{subfigure}[t]{0.45\linewidth}
        \centering
        \includegraphics[width=\linewidth]{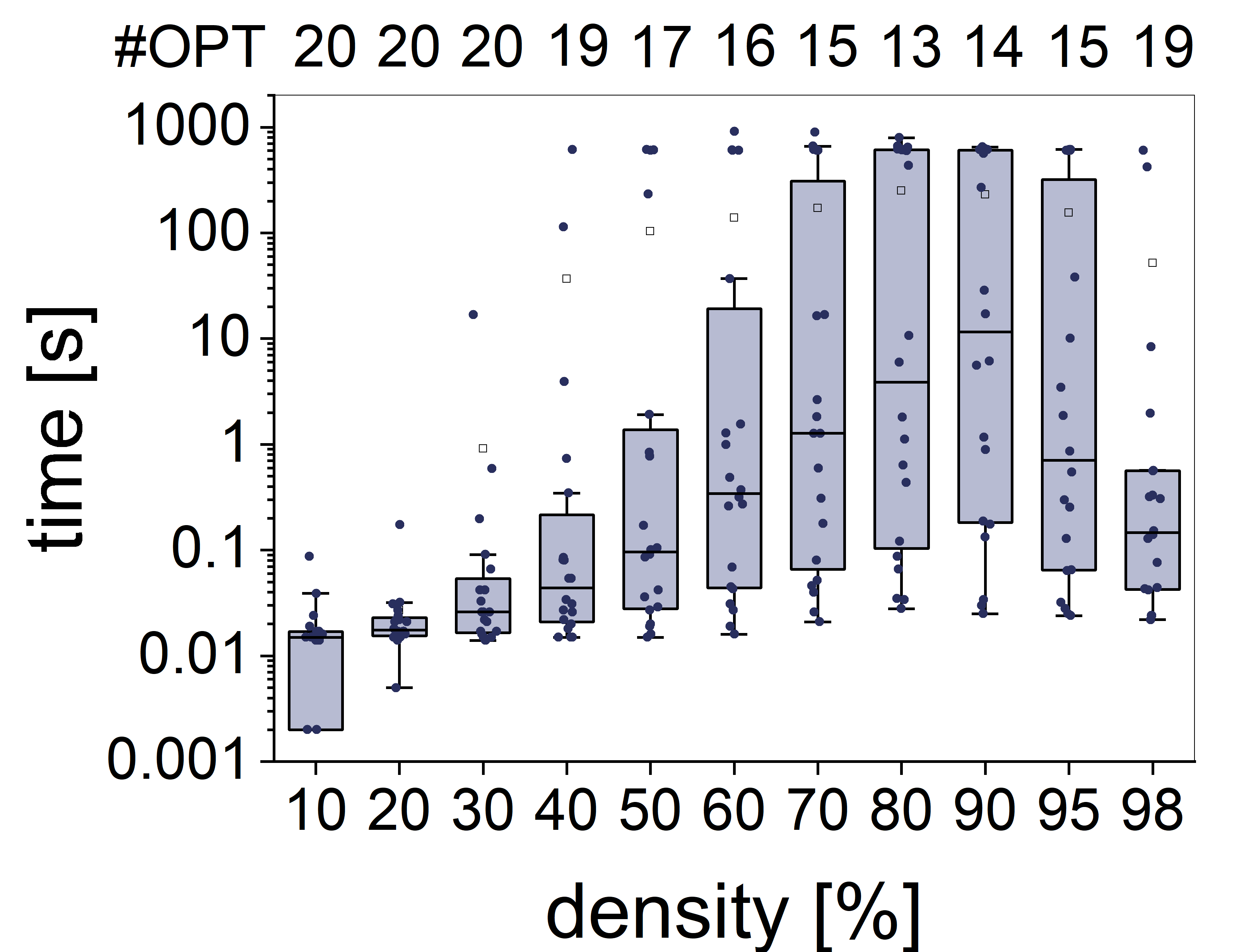}
        \caption{Runtime for each instance, with the number of completed instances within the 600-second time limit indicated at the top.}
        \label{fig:rand1}
    \end{subfigure}
    \hspace{0.02\linewidth}
    \begin{subfigure}[t]{0.45\linewidth}
        \centering
        \includegraphics[width=\linewidth]{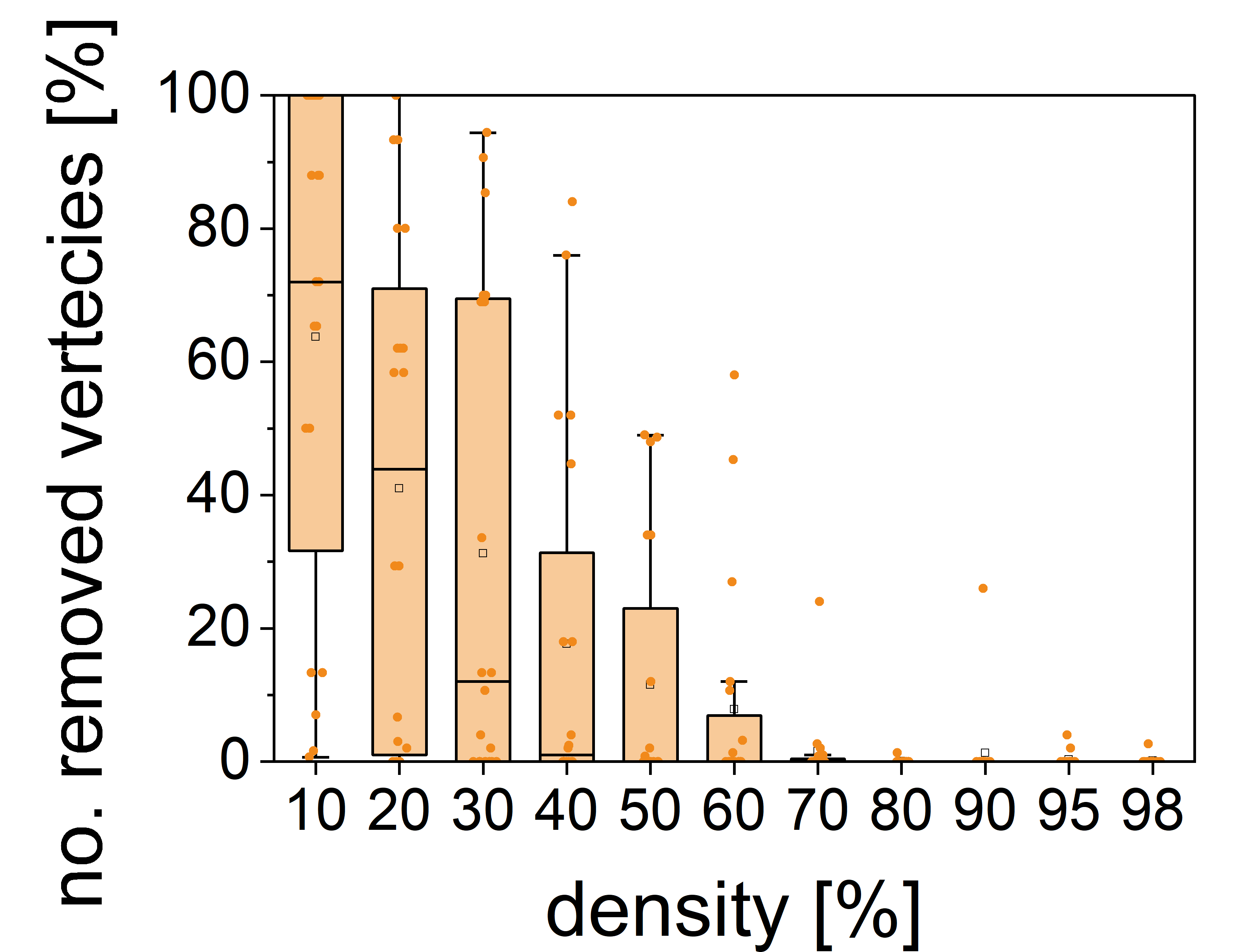}
        \caption{The percentage of vertices removed for each instance during the reduction process.}
        \label{fig:rand2}
    \end{subfigure}
    \caption{Results on random graphs, grouped by density.}
    \label{fig:rand}
\end{figure}
In order to know how graph density affects the reduction algorithm, in Figure~\ref{fig:rand}, we present the runtime and the percentage of vertices removed for different graph densities $p$.
The average runtime reaches its peak when $p$ is in $[0.8,0.9]$, and few vertices can be removed when $p\ge 0.7$.
We guess that in these dense graphs, the reduction rules in RECIP are less effective, making it perform worse than CLIQUE-INTER on such instances.  Additional evaluation on another dense benchmark from the Second DIMACS Implementation Challenge\footnote{This dataset can be downloaded from \url{https://iridia.ulb.ac.be/~fmascia/maximum_clique/}} confirms this trend, with RECIP and CLIQUE-INTER competing closely with each other.

\section{Conclusion and Future Work}

In this paper, we introduce new reduction rules and propose a novel algorithm, RECIP. Experiments demonstrate that our algorithm outperforms the current state-of-the-art method by an order of magnitude in terms of efficiency on large-scale datasets. The results highlight the effectiveness of our reduction techniques, with some limitations observed in dense graphs.
Thus, it is possible to further enhance its overall performance by exploring new reduction rules for dense graphs and improving the branch-and-cut algorithm in the future.
\section{Appendix}
\subsection{Missing Pseudo-code}
\begin{algorithm}[h]
    \caption{Color Reduction}
    \label{alg:graph_reduction}
    \textbf{Input}: graph \( G = (V, E) \) , an integer \( lb \), a coloring \( c \) of the graph \\
    \textbf{Output}: graph after reduction \( G'\) 
    \begin{algorithmic}[1]
        \FORALL{\( u \in V \)}
            \STATE Compute saturation degree \( ds_c(u) \)
        \ENDFOR
        \STATE Initialize an empty queue \( Q \)
        \FORALL{\( u \in V \) such that \( ds_c(u) \leq lb - 2 \)}
            \STATE Add \( u \) to \( Q \)
        \ENDFOR
        \WHILE{\( Q \) is not empty}
            \STATE \( x \gets \text{front of } Q \)
            \STATE Pop \( x \) from \( Q \)
            \IF{\( ds_c(x) \leq lb - 2 \)}
                \STATE Remove \( x \) from \( G \)
            \ELSE
                \STATE Update \( c(x) \) to the smallest valid color
            \ENDIF
            \FORALL{\( y \in N(x) \)}
                \STATE Update \( ds_c(y) \)
                \IF{\( ds_c(y) \leq lb - 2 \) \textbf{or} \( c(y) \) can be decreased}
                    \IF{\( y \notin Q \)}
                        \STATE Add \( y \) to \( Q \)
                    \ENDIF
                \ENDIF
            \ENDFOR
        \ENDWHILE
        \STATE \textbf{return} \( G \)
    \end{algorithmic}
\end{algorithm}
Algorithm~\ref{alg:graph_reduction} provides the detailed implementation of the color reduction step in the paper, enabling dynamic color modifications. Since each entry into the queue either removes a vertex or decreases the color of a vertex, the total time complexity is \( O(nm) \).

\begin{algorithm}[h]
    \caption{Get Disjoint Cliques}
    \label{alg:get_disjoint_cliques}
    \textbf{Input}: Graph \( G = (V, E) \)\\
    \textbf{Output}: a set of disjoint cliques \( \mathcal{C} \) 
    \begin{algorithmic}[1]
        \STATE \( \mathcal{C} \gets \emptyset \)
            \FORALL{\( x \in V \)}
                \STATE \(B=true\)
                \FORALL{\(C\in \mathcal{C}\)}
                \IF{\( x \) is connected to all vertices in \( C \)}
                    \STATE \( C \gets C \cup \{x\} \)
                    \STATE \(B=false\)
                    \STATE break;
                \ENDIF
                \ENDFOR
                \IF{\(B\)}
                   \STATE \( \mathcal{C} \gets \mathcal{C} \cup \{\{x\}\} \) 
                \ENDIF
            \ENDFOR
        \STATE \textbf{return} \( \mathcal{C} \)
    \end{algorithmic}
\end{algorithm}
Algorithm~\ref{alg:get_disjoint_cliques} describes the process of finding disjoint cliques as part of the lower bound computation in the paper. To generate different types of disjoint cliques, the order in which vertices in \( V \) are accessed can be modified. The total time complexity of this algorithm is \( O(n^2) \).

\subsection{Missing Proofs}
\begin{lemma}[Exact Clique Reduction]
\label{lemma_clique}
Given a graph $G=(V,E)$, an integer $k$, and an integer $lb$, if $\theta(G,k)\ge lb$ and there is a vertex $u\in V$ such that $\omega(G[N(u)]) \leq lb-2$, then $\theta(G,k)=\theta(G[V\setminus \{u\}])$. 
\end{lemma}
\begin{proof}

Given a graph $G=(V,E)$, an integer $k$, and an integer $lb\leq\theta(G,k)$. There is a vertex $u\in V$ such that $\omega(G[N(u)]) \leq lb-2$.

Let \( G' = G[V \setminus \{u\}] \) and \( S \in \mathcal{S}(G', k) \). Since \( G' \) is a subgraph of \( G \), it follows that \( \theta(G, k) \geq \theta(G', k) \geq \theta(G,k) -1 \). 
Additionally, because \( S \in \mathcal{S}(G', k) \), we have \( \omega(G[V \setminus \{u\} \setminus S]) = \theta(G', k) \). Due to \( \omega(G[N(u)\setminus S)])\leq\omega(G[N(u)]) \leq lb-2 \leq \theta(G,k)-2\leq\theta(G', k)-1\), we have \( \omega(G[V \setminus S]) = \omega(G[V \setminus \{u\} \setminus S]) = \theta(G', k) \). As a result, \( \theta(G, k) \leq \omega(G[V \setminus S]) = \theta(G', k) \). Since \( \theta(G, k) \geq \theta(G', k) \) was already established, we conclude that \(\theta(G, k) = \theta(G', k)\).
\end{proof}

\begin{lemma}
    Given a graph $G=(V,E)$, an integer $k$, and two vertices $u,v \in V$, if $\omega(G[N(u)]) <\omega(G[N(v)])$, then $\omega(G[N(v)\setminus \{u\}])=\omega(G[N(v)])$.
\end{lemma}

\begin{proof}

Given a graph $G=(V,E)$ and an integer $k$. There are two vertices $u,v \in V$ such that $\omega(G[N(u)]) <\omega(G[N(v)])$.

Let \( C \subseteq N(v) \) be a clique with \( |C| = \omega(G[N(v)]) \). Assume for contradiction that \( \omega(G[N(v) \setminus \{u\}]) \neq \omega(G[N(v)]) \). This implies \( u \in C \). 

Now consider the clique \( C' = (C \setminus \{u\}) \cup \{v\} \). We have \( C' \subseteq N(u) \) and \( C' \) is also a clique, so \(|C'| \leq \omega(G[N(u)])\). However, the size of \( C' \) is $|C'| = |C| - 1 + 1 = \omega(G[N(v)]) > \omega(G[N(u)]).$ and this is contradictory to \(|C'| \leq \omega(G[N(u)])\).
Therefore, \( \omega(G[N(v) \setminus \{u\}]) = \omega(G[N(v)]) \) must hold.
\end{proof}

\begin{lemma}[Triangle Reduction]
Given a graph $G=(V,E)$, an integer $k$ and an integer $lb$, if $ \theta(G,k) \geq lb$ and there is an edge $\{u,v\}\in E$ such that $|N(u)\cap N(v)|  \leq lb - 3$  then $\theta(G,k)=\theta((V,E\setminus \{\{u,v\}\}),k)$. 
\end{lemma}

\begin{proof}
Given a graph $G=(V,E)$, an integer $k$, and an integer $lb\leq\theta(G,k)$. There is an edge $\{u,v\}\in E$ such that $|N(u)\cap N(v)|  \leq lb - 3$.

Let \( G' = (V,E\setminus \{\{u,v\}\}) \) and \( S \in \mathcal{S}(G', k) \). Since \( G' \) is a subgraph of \( G \), it follows that \( \theta(G, k) \geq \theta(G’, k) \geq \theta(G,k) -1 \). 

Additionally, because \( S \in \mathcal{S}(G', k) \), we have \( \omega(G'[V \setminus S]) = \theta(G', k) \). We also have $\omega(G[V\setminus S]) \geq \theta(G,k) \geq lb$. Due to $\omega(G[(N(u)\cap N(v))\setminus S])+2\leq\omega(G[N(u)\cap N(v)])+2\leq|N(u)+N(v)|+2\leq lb$, we have $\omega(G[V\setminus S])=\omega(G'[V\setminus S]) = \theta(G', k)$
As a result, \( \theta(G, k) \leq \omega(G[V \setminus S]) = \theta(G', k) \). 
Since \( \theta(G, k) \geq \theta(G', k) \) was already established, we conclude that \(\theta(G, k) = \theta(G', k)\).
\end{proof}

\begin{lemma}[Interdiction Reduction]
Given a graph \( G = (V, E) \), an integer \( k > 0\), and a vertex \( u \in V \), if for any $v \in V\setminus N[u]$, \( \omega(G[N(u)]) - k > \omega(G[N(v)]) \), then \( \theta(G, k) = \theta(G[V \setminus \{u\}], k - 1) \).  
\end{lemma}

\begin{proof}
Given a graph \( G = (V, E) \) and an integer \( k > 0 \),and a vertex \( u \in V \) such that for every \( v \in V \setminus N[u],
\omega(G[N(u)]) - k > \omega(G[N(v)]).\)

Let \( S \in \mathcal{S}(G, k) \) and \( w \in S \). It follows that \( \omega(G[V \setminus S]) = \theta(G, k) \).

Now, if \( u \notin S \), since for any \( v \in V \setminus N[u] \), \( \omega(G[N(u)]) - k > \omega(G[N(v)]) \), we have\(\omega(G[V \setminus S]) = \omega(G[V \setminus (S \cup \{u\})]) + 1 \geq \omega(G[V \setminus ((S \setminus \{w\}) \cup \{u\})]).\)Let \( S' = (S \setminus \{w\}) \cup \{u\} \). We observe that \( |S'| = |S| \), thus \( |S'| \leq k \), and we also have\(\omega(G[V \setminus S']) \leq \omega(G[V \setminus S]) = \theta(G, k).\)Therefore, \( S' \in \mathcal{S}(G, k) \).

Thus, there exists \( S'' \in \mathcal{S}(G, k) \) such that \( u \in S'' \). As a result, we conclude that\(\theta(G, k) = \theta(G[V \setminus \{u\}], k - 1).\)




\end{proof}

\begin{lemma}[Domination Reduction]
Given a graph \( G = (V, E) \), an integer \( k \), and two vertices \( u, v \in V \), 
if \( N(v) \subset N(u) \) or \( N[v] \subset N[u] \), then if there exists a set \( S \in \mathcal{S}(G, k) \) such that \( v \in S \) and \( u \notin S \), then \( (S \setminus \{v\}) \cup \{u\} \in \mathcal{S}(G, k) \).  
\end{lemma}
\begin{proof}
    Given a graph \( G = (V, E) \), an integer \( k \), and two vertices \( u, v \in V \), such that \( N(v) \subset N(u) \) or \( N[v] \subset N[u] \). There exists a \(S \in \mathcal{S}(G, k) \) such that \( v \in S \) and \( u \notin S \).
    It follows that \(\omega(G[V\setminus S])=\theta(G,k)\). Let \(S'=(S \setminus \{v\}) \cup \{u\}, G' = G[V\setminus S']\). Due to \(|(S \setminus \{v\}) \cup \{u\}=|S|\leq k\),  we have $\omega(G') \geq \theta(G,k)$.
    There exists a \(C \subseteq V\setminus S'\) is a clique and \(|C|=\omega(G')\)
    If $v\in C$, because \( N(v) \subset N(u) \) or \( N[v] \subset N[u] \), we denote $C'=(C\setminus\{v\})\cup\{u\}$. It follows that \(C'\) is a clique and $C'\subseteq V\setminus S$. We have $|C'|=|C|+1-1=\omega(G')$, thereby $\omega(G')\leq\omega(G[V\setminus S])=\theta(G,k)$.
    Otherwise, if $v\notin C$, thereby $C\subseteq V\setminus S$, so $\omega(G')\leq\omega(G[V\setminus S])=\theta(G,k)$.
    
    Since $\omega(G') \geq \theta(G,k)$ was already established, we conclude that \(\omega(G') = \theta(G,k)\) and \( (S \setminus \{v\}) \cup \{u\} \in \mathcal{S}(G, k) \).  
\end{proof}


\subsection{Experiment on dimacs2}
On the graphs with $|V|=200$ from DIMACS2, we compare the performance of the two algorithms under two different budget values (\( k \in \{20,40 \} \)) with a runtime limit of 600 seconds.

Table~\ref{tab:performance} shows that among all 32 test cases, both algorithms timed out (T.L.) in 12 cases, RECIP performed better in 9 cases, and CLIQUE-INTER was superior in 11 cases. Notably, there was one instance where RECIP timed out while CLIQUE-INTER did not. In all experiments, RECIP failed to reduce any vertices, further demonstrating its limitations on dense graphs.

\begin{table}[htbp]
  \centering
  \caption{Compare with CLIQUE-INTER on the graphs with $|V|=200$ from DIMACS2}
  \label{tab:performance}
  \begin{tabular}{lrrrrrr}
    \toprule
    \multirow{2}{*}{Graph} & \multirow{2}{*}{$|E|$} & \multirow{2}{*}{$\omega(G)$} & \multicolumn{2}{c}{RECIP(s)} & \multicolumn{2}{c}{CLIQUE-INTER(s)} \\
    \cmidrule(lr){4-5} \cmidrule(lr){6-7}
    & & & k=20 & k=40 & k=20 & k=40 \\
    \midrule
    brock200\_1.clq & 14834 & 21 & 519.998 & T.L. & 554.504 & T.L. \\
    brock200\_2.clq & 9876 & 12 & 0.685 & T.L. & 0.878 & T.L. \\
    brock200\_3.clq & 12048 & 15 & 2.711 & 113.830 & 1.139 & 165.138 \\
    brock200\_4.clq & 13089 & 17 & T.L. & T.L. & T.L. & T.L. \\
    c-fat200-1.clq & 1534 & 12 & 0.184 & 0.375 & 0.374 & 0.064 \\
    c-fat200-2.clq & 3235 & 24 & 0.070 & 0.319 & 0.335 & 0.026 \\
    c-fat200-5.clq & 8473 & 58 & 0.118 & 0.276 & 0.328 & 0.107 \\
    san200\_0.7\_1.clq & 13930 & 30 & 4.855 & 58.181 & 3.829 & 78.826 \\
    san200\_0.7\_2.clq & 13930 & 18 & 10.089 & 86.751 & 13.728 & 76.735 \\
    san200\_0.9\_1.clq & 17910 & 70 & 2.291 & 8.959 & 0.341 & 13.551 \\
    san200\_0.9\_2.clq & 17910 & 60 & 36.369 & T.L. & 2.184 & T.L. \\
    san200\_0.9\_3.clq & 17910 & 44 & T.L. & T.L. & T.L. & T.L. \\
    sanr200\_0.7.clq & 13868 & 18 & 43.684 & T.L. & 26.187 & T.L. \\
    sanr200\_0.9.clq & 17863 & 62 & T.L. & T.L. & T.L. & T.L. \\
    gen200\_p0.9\_44.clq & 17910 & 44 & T.L. & T.L. & 382.658 & T.L. \\
    gen200\_p0.9\_55.clq & 17910 & 55 & 250.194 & T.L. & 35.709 & T.L. \\
    \bottomrule
  \end{tabular}%
\end{table}%

\section*{Acknowledgment}
We thank the authors of CLIQUE-INTER for kindly providing the source code.
This work was supported by the National Natural Science Foundation of China under grant 62372093. 

\printbibliography[title=References]

\end{document}